\documentclass[journal,twoside]{IEEEtran}
\usepackage[cmex10]{amsmath}
\usepackage{amsfonts}
\usepackage{amsmath}
\usepackage{amssymb}
\usepackage{graphicx}
\usepackage{float}
\usepackage[mathscr]{eucal}
\usepackage{amsfonts}
\usepackage{float}
\usepackage[colorlinks,linkcolor=blue]{hyperref}
\usepackage[latin1]{inputenc}
\usepackage{amsmath}
\usepackage{graphicx}
\usepackage{amssymb}
\usepackage{amsthm}
\usepackage{verbatim}
\usepackage{epsfig}
\usepackage[labelfont=bf]{caption}
\usepackage{enumerate}
\usepackage{subfig}
\usepackage{epstopdf}
\newtheorem{lemma}{Lemma}
\newtheorem{proposition}{Proposition}
\newtheorem{corollary}{Corollary}
\newtheorem{fact}{Fact}

\newtheorem{remark}{Remark}

\newtheorem{assumption}{Assumption}

\def\begcen{\begin{center}}
\def\endcen{\end{center}}

\newcommand{\col}{ \mbox{col} }

\def\L2{{\cal L}_2}
\def\L2e{{\cal L}_{2e}}

\def\rea{\mathbb{R}}

\def\x{{x}}

\def\begequarr{\begin{eqnarray}}
\def\endequarr{\end{eqnarray}}
\def\begequarrs{\begin{eqnarray*}}
\def\endequarrs{\end{eqnarray*}}
\def\begarr{\begin{array}}
\def\endarr{\end{array}}
\def\begequ{\begin{equation}}
\def\endequ{\end{equation}}
\def\lab{\label}
\def\begdes{\begin{description}}
\def\enddes{\end{description}}
\def\begenu{\begin{enumerate}}
\def\begite{\begin{itemize}}
\def\endite{\end{itemize}}
\def\endenu{\end{enumerate}}

\def\lef[{\left[\begin{array}}
\def\rig]{\end{array}\right]}

\def\begcen{\begin{center}}
\def\endcen{\end{center}}
\def\begrem{\begin{remark}\rm}
\def\endrem{\end{remark}}
\def\begassum{\begin{assumption}}
\def\endassum{\end{assumption}}
\def\begassums{\begin{assumption*}}
\def\endassums{\end{assumption*}}
\def\begassu{\begin{ass}}
\def\endassu{\end{ass}}
\def\beglem{\begin{lemma}}
\def\endlem{\end{lemma}}
\def\begcor{\begin{corollary}}
\def\endcor{\end{corollary}}
\def\begfac{\begin{fact}}
\def\endfac{\end{fact}}
\def\TAC{{\it IEEE Trans. Automat. Contr.}}
\def\AUT{{\it Automatica}}

\def\SCL{{\it Systems and Control Letters}}

\def\L2e{{\cal L}_{2e}}
\def\L2{{\cal L}_{2}}

\def\rea{\mathbb{R}}
\def\intnum{\mathbb{Z}}

\def\col{\mbox{col}}

\def\IJRNLC{{\it Int. J. on Robust and Nonlinear Control}}
\def\TAC{{\it IEEE Trans. Automatic Control}}

\def\IJC{{\it International Journal of Control}}

\def\SCL{{\it Systems \& Control Letters}}
\def\AUT{{\it Automatica}}

\def\intnum{\mathbb{Z}}

\def\begsubequ{\begin{subequations}}
    \def\endsubequ{\end{subequations}}
\usepackage{xcolor}

\graphicspath{{figs/}}

\begin{document}

\title{Robust Integral Consensus Control of Multi-Agent Networks Perturbed by Matched and Unmatched Disturbances: The Case of Directed Graphs}

\author{Jose Guadalupe Romero,~\IEEEmembership{Member,~IEEE,} and David Navarro-Alarcon,~\IEEEmembership{Senior~Member,~IEEE}%
\thanks{This work is supported in part by the Research Grants Council (RGC) of Hong Kong under grant 15212721 and grant 15231023.}%
\thanks{Jose Guadalupe Romero is with the Department of Electrical and Electronic Engineering, Instituto Tecnologico Autonomo de Mexico (ITAM), Mexico City, Mexico. (email: jose.romerovelazquez@itam.mx).}
\thanks{David Navarro-Alarcon is with Department of Mechanical Engineering, The Hong Kong Polytechnic University (PolyU), Kowloon, Hong Kong. (e-meail: dnavar@polyu.edu.hk)}
}

\bstctlcite{IEEEexample:BSTcontrol}

\maketitle

\begin{abstract}
This work presents a new method to design consensus controllers for perturbed double integrator systems whose interconnection is described by a directed graph containing a rooted spanning tree. 
We propose new robust controllers to solve the consensus and synchronization problems when the systems are under the effects of matched and unmatched disturbances. 
In both problems, we present simple continuous controllers, whose integral actions allow us to handle the disturbances. 
A rigorous stability analysis based on Lyapunov's direct method for unperturbed networked systems is presented. 
To assess the performance of our result, a representative simulation study is presented.
\end{abstract}

\begin{IEEEkeywords}
				Multi-agent systems, consensus control, directed networks, matched and unmatched disturbances.
\end{IEEEkeywords}
\IEEEpeerreviewmaketitle

\section{Introduction}
\lab{sec1}
Due to the recent technological advances and affordability of consumer-level autonomous systems, the control community has paid considerable attention to various control problems in multi-agent systems. 
Some classical examples include the design of formation \cite{CHEWAN}, flocking \cite{OLF}, consensus control strategies \cite{RENBEA}. 
The consensus control problem has been of particular interest to researchers, since the computed {\it final agreement} of a network of agents represents a crucial task in many real-world applications, e.g., in robot fleets, electrical power networks, biological systems \cite{RENBEARD, MENDIMJOH, FERetal}, etc.

Graph theory is the main tool used in the analysis of consensus control problems, where the network's Laplacian matrix describes the interconnection and communication properties among the agents. 
An undirected graph (which models bidirectional communication) has associated a symmetric and positive semi-definite Laplacian matrix; This valuable property enables to use various well-known results to calculate the final agreement for general (non)linear systems, Euler-Lagrange dynamics, nonholonomic robots, among others \cite{Lietal11, REN09, LISTMASADA}. 
In these works, the presence of uncertainties, unknown parameters, input delays and disturbances have been typically tackled with adaptive or robust control techniques, see e.g., \cite{PHIetal, NUNOetal, SUetal, ROMNUNALD}. 
With undirected networks, it is relatively easy to conduct Lyapunov-based stability analysis, even in the presence of nonlinear time-varying uncertainties \cite{BECCROV, KATROV}.

When considering directed graphs, the consensus problem is significantly more complicated since the Laplacian matrix is no longer symmetric. 
This slight difference, yet with great implications, means that we can no longer use the same controller design techniques applied to the above-mentioned undirected networks.
It is worth noting that in real applications, a directed graph is the most natural and realistic way to model information exchange among a group of agents, since their communication is not necessarily bidirectional. 
This common situation arises due to the limited sensing capability of transducers, as well as their weak communication ranges and intermittent connectivity. 
Multi-agent networks where local unidirectional information exchange is allowed are more convenient in terms of cost, scalability and flexibility, however, their analysis and controller design presents many challenges.

When a directed graph is strongly connected and structurally unbalanced, the left eigenvector of the Laplacian matrix consists of positive elements; A Lyapunov function for the consensus problem can then be constructed using this eigenvector, see e.g., \cite{Zhanetal,  DASLEW, GHACOR, DOMHAD, CAOMORAND}. 
Robust consensus controllers have been proposed for directed networks considering various problems such as unmeasurable velocity, matched disturbances, delays, and other uncertainties \cite{SHIJOH, LIetal19, ZHANGetal}. 
However, the strong connectedness requirement used in these works is very restrictive, as it implies that every agent is reachable from every other agent in the network, which hard to satisfy in practice. 
To relax this condition, in \cite{RENBEA05} was proved that consensus can be established if the graph describing the interconnection has a rooted spanning tree, which is a considerably less restrictive situation.
Examples that use this approach include consensus for linear systems \cite{Lietal15}, second-order heterogeneous systems \cite{MEIRENCHEN},  uncertain multi-agent linear systems \cite{LIUWUZHO}, among others. 

A key problem in multi-agent consensus is to design effective control strategies that can deal with unknown matched and unmatched disturbances to the input. 
For the fist case (i.e., when disturbances appears in the control input channel), several works have been published considering linear multi-agent systems \cite{LVetal17}, finite-time consensus for second-order systems \cite{ZHAOetal17}, unknown velocities \cite{TIALIULIU}, consensus with disturbances generated by a known linear integrator \cite{SUNetal}, and for disturbances generated by exosystems \cite{Yuetal23}. 
Although in some works the local disturbance rejection has been proved as well as the consensus goal, {\it all them} are based on complex designs, namely, they rely on {\it discontinuous} or {\it high-gain} adaptive observers combined with discontinuous (terminal) sliding mode controllers. 
It is well known that this type controllers exhibit robustness against matched uncertainties, however, the main disadvantage is that they lead to control signals that may produced undesired chattering effects on the actuators \cite{EDWSPU}. 
 
As unmatched disturbances do not appear in the control input channel, their active compensation presents many challenges.
This type of disturbances are common in many systems, e.g., in mechanical systems when the {\it velocity} measurements are corrupted \cite{ROMDONORT}, in missile guidance systems due to torques arising from external wind or due to variation of aerodynamic coefficients \cite{CHEN}; These disturbances are also common in power electronics like the DC-DC and DC-AC buck and {\^C}uk converters \cite{SUNetal, BATFOSOLI}. 
When unmatched disturbances are present in a network with a directed graph, only {\it partial} consensus or synchronization can be established \cite{SUNetal, WANGetal21}. 
For double-integrator systems, the portion of the state variables that reach consensus correspond to the unactuated variables, typically referred as the {\it output} state.

Various works have addressed this unmatched disturbance case, e.g.,  a controller to ensure output consensus under a strongly connected graph was presented in \cite{WANGetal21}. 
In \cite{GUetal21}, robust output consensus {\it tracking} was guaranteed in finite time and considering a directed graph with spanning tree. 
A time varying adaptive {\it output} formation control scheme for collision avoidance via artificial potential for second-order systems with both matched and unmatched disturbances was proposed in \cite{ZHENGetal23}. 
A neural network based adaptive containment controller was presented in \cite{XIAOetal22}; The work proved that the proposed containment algorithm ensures that the closed loop systems are finite-time stable and containment errors converge to a small residual set around the origin. 
However, all these previous works are based on dynamic gains and {\it discontinuous} adaptive observers/controllers, which may yield undesired effects in the control signal.
Recently, in \cite{PANLORSUK} was presented a strict Lyapunov function for {\it dynamic} consensus of systems of networks with a directed spanning tree \cite{LIDUACHE, OLFMUR}; This work (which is based on \cite{PANLOR17}) provides a constructive proof for global exponential stability, which is ensured under simple conditions of the control and Lyapunov gains.
Some new results have also explored this dynamic consensus idea, e.g., for linear systems \cite{DUTetal22}, and for model reference adaptive control \cite{DUTetal23, DUTetal23a}.

In this paper, we address the robust controller design for multi-agent systems perturbed by constant matched and unmatched disturbances, and whose interconnection is described by a directed graph. 
In contrast with existing solutions, our proposed method uses a {\it simple} and {\it smooth} integral action to deal with disturbances.
Since complex solutions have been presented to ensure the consensus of the called {\it output state} when unmatched disturbances are presented, we relax the solution to the synchronization of periodic (i.e., closed) orbits \cite{SCASEP, ORTetal20}, which is a more frequently encountered problem  in many applications, e.g., in power systems \cite{BATFOSOLI, SAIISHISH, ANGOLITAB}. 

The original contributions of this work are listed as follows:
\begin{itemize}
    \item We propose a new control scheme to ensure dynamic consensus for perturbed multi-agent systems with directed communication.
    \item We propose a new integral action to reject constant matched disturbances, and for the case of unmatched disturbances, to ensure synchronization of periodic orbits.
    \item We propose a new strict Lyapunov function to rigorously analyze the stability of our smooth integral controller.
    \item We report a detailed numerical study to validate the performance of our proposed method.
\end{itemize}
  
The rest of the paper is organized as follows: Section II presents   preliminaries and assumptions  to be used; Section III contains our main result; and, finally, the simulations and conclusions are shown in Sections IV and V, respectively. 
\paragraph{Notation.} $\rea_{>0}$, $\rea_{\geq 0}$, $\intnum_{>0}$ and $\intnum_{\geq 0}$ denote the positive and non-negative real and integer numbers, respectively. $\| x\|$ stands for the standard Euclidean norm of vector $x\in\rea^n$. $I_n$ represents the identity matrix of size $n\times n$. ${\bf1}_k$ stands for a column vector of size $k$ with all entries equal to one.  The set ${\bar N}$ is defined as ${\bar N}:=\{1,\dots,N\}$, where $N$ is a positive natural number.

\section{Problem Formulation}
The interconnection graph between the agents may be modelled by a constant Laplacian matrix, ${ L}:=[\ell_{ij}] \in{\mathbb{R}}^{N\times N}$, whose $i$th-$j$th element satisfies: 
\begin{equation}
\label{lap:eq}
{\ell _{ij}} = \left\{ {\begin{array}{*{20}{c}}
{\sum\limits_{k \in  \mathcal N_i} {{a_{ik}}} }&{i = j}\\
{ - {a_{ij}}}&{i \ne j},
\end{array}} \right.
\end{equation}
where $ \mathcal N_i\subset \mathbb{Z}$ is the set of indices corresponding to agents that locally transmit information to the $i$th agent, $a_{ij}>0$ denotes the connectivity between agents in the network (no self connections are considered, thus, $a_{ij}=0$). 
For directed graphs, the Lapacian matrix is typically not symmetric, i.e., $a_{ij} \not= a_{ji}$ is generally satisfied.
In this work, we make the following key assumption:
\begin{assumption}\em
\label{asp1}
The directed graph that models the interactions among agents in the network contains a directed spanning tree.
\end{assumption}

Based on this assumption, the following Lemmata hold:

\begin{lemma} \em
\label{lem1}
\cite{RENBEARD} The Laplacian matrix $L$ has a unique zero-eigenvalue and, by construction, the rest of its spectrum is strictly positive and satisfies $L \boldsymbol 1_N = 0$, with $\boldsymbol 1_N \in \rea^N$ as its associated right eigenvector. 
\end{lemma}

\begin{lemma} \em
\label{lem2}
\cite{PANLORSUK} Let ${\mathcal G}$ be a directed graph and $L$ its corresponding non-symmetric Laplacian matrix. 
Then, for any positive matrix $Q>0 \in \rea^{N \times N}$ and scalar $\alpha>0$, there exists a positive symmetric matrix $P>0 \in \rea^{N \times N}$ such that:
\begin{align}
PL+L^\top P= Q-\alpha [P {\boldsymbol 1}_N  v_\ell^\top +v_\ell  {\boldsymbol 1}_N^\top P],
\end{align}
  where the column vector $v_\ell$ denotes the left eigenvector associated with the single zero eigenvalue of $L$.
\end{lemma}


\paragraph{Problem statement.} Consider a group of double-integrator linear systems subject to matched and unmatched constant disturbances, and whose interconnection satisfies Assumption \ref{asp1}.
For this class of dynamic systems, we aim to design an integral controller that  (i) can reject all matched disturbances and thus ensure the consensus of all agents and (ii) ensures synchronization of the output state of all agents to a periodic orbit when constant unmatched disturbances are present.

\section{Main Result}
\lab{sec2}

\subsection*{Matched Disturbances}
  
  The multi-agent system to be addressed is of the form:
  \begin{align}
  \label{sys2}
  \dot  x_i =& \; y_i \quad x_i, y_i\in \rea,\;  i \leq N \nonumber \\
  \dot y_i=& \; u_i +d_i 
  \end{align}
  with $u_i \in \rea$ as the input control and $d_i \in \rea$ the external constant disturbance.
  The interconnection among the agents is assumed to satisfy Assumption \ref{asp1}.
  
  Our goal is to design a robust controller such that dynamic consensus can be established, i.e.: 
 \begin{align}
 \label{goal1}
\lim_{t\to \infty} &\quad |x_{i}(t)-x_{m}(t) |=0, \;\; |y_{i}(t)-y_{m}(t) |  =0 \nonumber \\
\lim_{t\to \infty} &\quad |\tilde \delta_{i}(t)-\delta_m(t) |=0 , \quad i \not=m.
\end{align}
for a disturbance error $\tilde \delta_i= \hat \delta -\frac{d_i}{\gamma_3}$, with $\hat \delta (t)$ as an integral action whose aim is to compensate the disturbance, and $\gamma_3$ a free positive gain.  
This consensus implies the so-called {\it mean-field} dynamics \cite{PANLOR17}:
 \begin{align}
 \label{mf}
 x_{i}= v_\ell^\top x, \quad y_{i}= v_\ell^\top y, \quad \delta_i=v_\ell^\top \tilde \delta
 \end{align}
  which can be seen as a weighted average of the system.   
 Here, we have used the compact form $x= [x_1, \cdots, x_N]^\top$,  $y= [y_1, \cdots, y_N]^\top$, $ \tilde \delta = [\delta_1, \cdots, \tilde \delta_N]^\top$.

\begin{proposition}   \em
\label{pro1}
  Consider the system \eqref{sys2} in closed-loop with the following controller: 
  \begin{align}
u_i= &\; -\gamma_1 \sum\limits_{j=1}^{\mathcal N} a_{ij} (x_i-x_j)-\gamma_2 y_i  -\gamma_3 \hat \delta_i \nonumber  \\
\dot {\hat \delta}_i= &\; \gamma_1 \sum\limits_{j=1}^{\mathcal N} a_{ij} (x_i-x_j)+\gamma_4 y_i, 
\label{intc2}
\end{align}
where  $\gamma_1$, $\gamma_2$, $\gamma_3$ and $\gamma_4$ are positive gains to be defined, and $\hat \delta_i$ is the integral action whose aim is to eliminate the disturbance estimation error $\tilde \delta_i = \hat \delta_i-\frac{d_i}{\gamma_3}$. 
From the definition of the signals $x$, $y$, $\tilde \delta$
and noting that $\frac{\rm d}{{\rm d}t}{\hat\delta}_i = \frac{\rm d}{{\rm d}t} {\tilde\delta}_i$, we can express the closed-loop dynamics in the following compact form:
\begin{align}
\label{tilsys2}
\dot x=& \; y \nonumber \\
\dot y =& \; -\gamma_1 L x-\gamma_2 y - \gamma_3 \tilde \delta \nonumber \\
\dot {\tilde \delta}= & \; \gamma_1 L x +\gamma_4 y
\end{align}
with the following consensus errors\footnote{It is also referred to as synchronization errors in \cite{PANLOR17}.}
\begin{align}
\label{err2}
e_x=&\;(I_n-{\bf 1} v_\ell^\top) x, \quad e_y=\;(I_n-{\bf 1} v_\ell^\top) y \nonumber \\
e_d=&\;(I_n-{\bf 1} v_\ell^\top) \tilde \delta .
\end{align}  
which define the difference between the individual system states and the mean-field state.
This closed-loop system ensures the convergence to zero of $(e_x, e_y, e_d)$ $\to 0$, with a Lyapunov function defined as ${\mathcal H}= {\mathcal H}_s + {\mathcal H}_d$, where: 
 \begin{align}
 {\mathcal H}_s =& \; \frac{1}{2}  \left[ \begin{array}{cc} e_x^\top  & e_y^\top  \end{array} \right]\left[ \begin{array}{cc} \rho P &  \epsilon P \\ \epsilon P & 2 \mu I_n \end{array} \right]   \left[ \begin{array}{c} e_x \\ e_y  \end{array} \right], \nonumber \\
 {\mathcal H}_d=&\;  \frac{1}{2}b  \left[ \begin{array}{cc} e_y^\top  & e_d^\top  \end{array} \right]\left[ \begin{array}{cc} 2I_n &  I_n \\ I_n &  I_n \end{array} \right]   \left[ \begin{array}{c} e_y \\ e_d  \end{array} \right], 
 \end{align}
for positive scalar parameters $\rho, \epsilon, \mu, b > 0$ and a positive-definite symmetric matrix $P=P^\top > 0$, defined such that the condition $\sqrt{\frac{2 \rho \mu}{\|P\|}}> \epsilon$ is satisfied, and hence, ${\mathcal H}>0$. \end{proposition}
 
 \begin{proof}
 By computing the time derivatives of $e_x$ and $e_y$, we obtain the following dynamic equations:
\begin{align}
 \dot e_x=& \; (I_n-{\bf 1}_N v_\ell^\top)y= e_y \nonumber \\
 \dot e_y=&\; - (I_n-{\bf 1}_N v_\ell^\top) (\gamma_1 L x +\gamma_2 y+ \gamma_3 \tilde \delta)\nonumber \\
 =& \;-\gamma_1 (I_n-{\bf 1}_N v_\ell^\top)L(e_x+{\bf 1}_Nv_\ell^\top x) -\gamma_2 e_y-\gamma_3 e_d \nonumber \\
 =&\; -\gamma_1 L e_x -\gamma_2e_y -\gamma_3 e_d, 
 \label{error_traj}
\end{align}
 where to get the last equality we used the facts   $L {\bf 1}_N=0$ and  $v^\top_\ell L=0$. 
 The time derivative of  ${\mathcal H}_s$ along the trajectories \eqref{error_traj} is: 
 \begin{align}
 \dot{\mathcal H}_s&=  \rho e_x^\top P \dot e_x +2 \mu e_y^\top \dot e_y +\epsilon e_x^\top P \dot e_y +\epsilon \dot e_x^\top P e_y \nonumber \\
 &= \rho e_x^\top P e_y +2 \mu e_y^\top \left(-\gamma_1 L e_x -\gamma_2e_y -\gamma_3 e_d  \right) \nonumber \\
 &+ \epsilon e_x^\top P\left( -\gamma_1 L e_x -\gamma_2e_y -\gamma_3 e_d \right)+ \epsilon e_y^\top P e_y \nonumber \\
 &= e_x^\top \left( \rho P -2\gamma_1 \mu L^\top -\epsilon \gamma_2 P   \right) e_y -e_y^\top \left( 2 \mu  \gamma_2 I_n - \epsilon P \right)e_y\nonumber \\
 & -\tfrac{1}{2}\gamma_1\epsilon e_x^\top (PL+L^\top P) e_x -\epsilon \gamma_3 e_x^\top P e_d -2\mu \gamma_3e_y^\top e_d \nonumber \\
 &= e_x^\top \left( \rho P -2\gamma_1 \mu L^\top -\epsilon \gamma_2 P   \right) e_y -e_y^\top \left( 2 \mu  \gamma_2 I_n - \epsilon P \right)e_y \nonumber \\
 &- \tfrac{\epsilon}{2}\gamma_1e_x^\top e_x -\epsilon \gamma_3 e_x^\top P e_d -2\mu \gamma_3e_y^\top e_d,
 \label{dotHs}
  \end{align}
  where to get the last equality we have invoked Lemma \ref{lem2} with $Q=I_n$.  
  The time derivative of $\mathcal H_d$ satisfies the following:
\begin{align}
\dot {\mathcal H}_d =& \;b\left(e_y^\top (2 \dot e_y +\dot e_d)+ e_d^\top (\dot e_y +\dot e_d)\right) \nonumber \\
 =&\;  \;b\Big(e_y^\top \left\{-2 \Big [ \gamma_1 L e_x +\gamma_2e_y +\gamma_3 e_d\Big]+\dot e_d\right\} \nonumber \\
 &+ e_d^\top \left\{- \Big [ \gamma_1 L e_x +\gamma_2e_y +\gamma_3 e_d\Big] +\dot e_d \right\}\Big), \nonumber 
\end{align} 
where the term $\dot e_d$ is computed by making use of the definition of the integral action \eqref{intc2} as follows:
\begin{align}
\dot e_d=&\; (I_n-{\bf 1} v_\ell^\top) \dot{\tilde \delta}:=  (I_n-{\bf 1} v_\ell^\top)  ( \gamma_1 L x +\gamma_4 y) \nonumber \\
=&\;  \gamma_1  (I_n-{\bf 1} v_\ell^\top) L (e_x +{\bf 1} v_\ell x)+\gamma_4 e_y \nonumber \\
=&\; \gamma_1 L e_x +\gamma_4 e_y.
\end{align}
Replacing $\dot e_d$ into $\dot {\mathcal H}_d$ yields: 
\begin{align}
\dot {\mathcal H}_d =& \;b\left(e_y^\top (2 \dot e_y +\dot e_d)+ e_d^\top (\dot e_y +\dot e_d)\right) \nonumber \\
 =&\;  \;b\Big(e_y^\top \left\{- \gamma_1 L e_x -(2  \gamma_2-\gamma_4) e_y -2 \gamma_3 e_d \right\} \nonumber \\
 &+ e_d^\top \left\{- \gamma_2e_y - \gamma_3 e_d + \gamma_4 e_y \right\}\Big) \nonumber \\
 =&\; b\Big(- \gamma_1 e_y^\top  L e_x - e_y^\top (2  \gamma_2-\gamma_4) e_y \nonumber \\
 & -\gamma_3 e_d^\top e_d   + (\gamma_4 -2 \gamma_3  - \gamma_2 ) e_d^\top   e_y \Big).
 \label{dotHd}
\end{align} 
By using \eqref{dotHs} and \eqref{dotHd}, we can express $\dot{\mathcal H}$ in the compact form $\dot {\mathcal H}=-\frac{1}{2}e^\top {\mathcal N} e$, with an extended error vector $e=[e_x, \; e_y, \; e_d]^\top \in \rea^{3n}$ and a symmetric matrix defined as:
\begin{equation}
{\mathcal N}=  \left[ \begin{array}{ccc} \gamma_1 \epsilon I_n &    {\mathcal N}_{12}  & \gamma_3 \epsilon P \\ \star   & {\mathcal N}_{22}  & {\mathcal N}_{23}\\
\star &   \star & 2 \gamma_3 b I_n
\end{array} \right]. 
\end{equation}
with
\begin{align}
{\mathcal N}_{12}=&\; \gamma_1 (2 \mu +b) L^\top -(\rho -\epsilon \gamma_2) P   \nonumber \\
{\mathcal N}_{22}=& \;2 (2b  \gamma_2I_n- b \gamma_4 I_n+ 2 \mu  \gamma_2 I_n - \epsilon P) \nonumber \\
{\mathcal N}_{23}=&\;  I_n( 2\mu \gamma_3  +2b \gamma_3 +b \gamma_2-b\gamma_4).
\end{align}

For ease of presentation, let us introduce the following scalar parameters: 
\begin{align}
\gamma_4=2 \gamma_3\left(1 +\frac{\mu}{b} \right) +\gamma_2, \quad \epsilon = \frac{\rho}{\gamma_2},
\end{align}
which we can use to equivalently express the matrix ${\mathcal N}$ as: 
{\small
\begin{align}
\hspace{-.3cm}
{\mathcal N}=  \left[ \begin{array}{ccc}  \rho \frac{\gamma_1}{\gamma_2} I_n &    \gamma_1 (2 \mu +b) L^\top  & \gamma_3 \epsilon P \\ \star  & 2\gamma_2 (2 \mu +b) -2\frac{\rho}{\gamma_2} P- 4\gamma_3 (\mu+b) & 0 \\
\star &   0 & 2 \gamma_3 b I_n
\end{array} \right]. 
\nonumber
\end{align}
}

The stability proof of the system relies on the positive definitiveness of ${\mathcal N}$. Since we have two free Lyapunov parameters  $\rho$ and $b$ and several free control gains,  we apply the Schur complement to prove  that ${\mathcal N}>0$.  
A simple solution can be obtained by setting $\rho= \gamma_2$, then applying the Schur complement to the $2\times 2$ sub-block of ${\mathcal N}$ as:
$$
 \left[ \begin{array}{ccc}   \gamma_1 I_n &    \gamma_1 (2 \mu +b) L^\top \\ \gamma_1 (2 \mu +b) L  & 2\gamma_2 (2 \mu +b) -2 P- 4\gamma_3 (\mu+b) 
\end{array} \right]
$$
which yields the relation: 
\begin{align}
2\gamma_2 (2 \mu +b) -2P- 4\gamma_3 (\mu+b)-\gamma_1 (2 \mu+b)^2 L L^\top >0. 
\end{align}
By defining $\gamma_2$ as:
$$
\gamma_2 > \frac{1}{(2\mu+b)} \left( \lambda_p+2\gamma_3 (\mu+b)  \right)+\frac{1}{2}\gamma_1 (2 \mu+b) \lambda_L^2 
$$
with  $ |P| \leq \lambda_p $ and  $|L| \leq \lambda_L$, we can ensure that  $e_{xy}^\top {\mathcal N}_{22} e_{xy}\geq   |{\mathcal N}_{22} |  | e_{xy}|^2$, with $e_{xy}= \col(e_x, e_y)$.  
Finally, the positive definiteness of ${\mathcal N}$ is established with 
$$
b \geq \frac{\gamma_3}{\gamma_1} \lambda_p^2.
$$  
This ensures that $\lim_{t \to \infty} e(t)=0$, as consequence from \eqref{err2} and the mean field dynamics \eqref{mf} we have that \eqref{goal1} holds.
This completes the proof.
 \end{proof}
 
Now, the exact estimation of the input disturbances and the final states of the agents is presented in the following Corollary. Hence, we study the dynamic behavior of each agent, which is governed by the weighted average dynamics. 

\begin{corollary} \em
\label{cor1}
Consider the  mean-field coordinates \eqref{mf}, and assume that the positive gains  $\gamma_j$ with $j=2:4$ are such that: 
\begin{equation}
\label{matS}
{\mathcal S}=  \left[ \begin{array}{cc}      -\gamma_2  & -\gamma_3  \\ 
\gamma_4 & 0
\end{array} \right]
\nonumber
\end{equation}
 is a Hurwitz matrix. Then, each agent of the closed-loop system \eqref{tilsys2} satisfies:
\begin{align}
\lim_{t \to \infty}  x_{i}(t)  = x_{m}(0) +c_1
\end{align}
with $c_1$ as a positive constant, and the state variables
\begin{align}
\label{xd2}
y_{i} (t) \to 0, \quad  \tilde \delta  \to 0
\end{align}
exponentially converging to zero. As consequence:
\begin{align}
\lim_{t \to \infty} \hat \delta_i (t)= \frac{d_i}{\gamma_3}
\end{align}
guarantees the exact estimation of the disturbances.
\end{corollary}

\begin{proof}
The time derivative of \eqref{mf} along of the closed loop  \eqref{tilsys2} yields the following averaged model, which corresponds to {\it dynamic consensus}:
\begin{align}
\dot x_{m}=&\; v_\ell^\top y:= y_{m} \nonumber \\
\dot y_{m}= &- v_\ell^\top \left( \gamma_1 L x -\gamma_2 y - \gamma_3 \tilde \delta \right):= - \gamma_2 y_{m} - \gamma_3 \delta_m \nonumber \\
\dot \delta_m=&\; v_\ell^\top (\gamma_1 L x + \gamma_4 y):= \gamma_4y_{m}.
\label{dynav}
\end{align}
Where we note that the last two dynamic equations can be re-written as $\dot y_{\delta m}= {\mathcal S}  y_{\delta m}$ with $y_{\delta m}= \col(y_{m}, \delta_m) \in \rea^2$ and ${\mathcal S}$ as \eqref{matS}. 
This way, its unique solution can be computed as follows: 
\begin{equation}
\label{expx2d}
y_{\delta m}(t)=\exp^{{\mathcal S}t}y_{\delta m}(0).
\end{equation}
Since we assume that ${\mathcal S}$ is a Hurwitz, we invoke the Cayley-Hamilton theorem to establish \eqref{xd2}. 

On the other hand, the solution of the first equation of \eqref{dynav}  is 
\begin{align}
x_{m}(t)=&\; x_{m}(0)+\int_{0}^t y_{m}(s)ds \nonumber \\
=& \; x_{m}(0) +\int_{0}^t \exp^{-\frac{1}{c_1} s} ds :=  \; x_{m} (0)+c_1 \nonumber \\
\end{align}
where to get the second equality we have used the fact that $y_{\delta m}(t)$ converge exponentially with some constant $\frac{1}{c_1}>0$.

Finally, from Proposition \ref{pro1} is ensured that $x_{i}$ converges to $x_{m}$ which is  a {\it constant} and  from \eqref{expx2d} we have that $y_{i}$ and $\tilde \delta$ converge exponentially to zero,  then from the closed loop of each agent given by 
$$
\dot x_{i}= -\gamma_1 \sum\limits_{j=1}^{\mathcal N} a_{ij} (x_{i}-x_{j})-\gamma_2 y_{i}  -\gamma_3 \hat \delta_i  +d_i 
$$
we conclude that $\lim_{t \to \infty} \hat \delta_i = \frac{d_i}{\gamma_3}$ and exact estimation of the disturbance is guaranteed.
This completes the proof.
 \end{proof}
 
%
Next, we present the second main result of the note, i.e., the case where the multi-agent system is subject to unmatched disturbances. 
In this situation, the unmatched disturbances generate a bias on the un-actuated channel (i.e., at the $\dot x_i$ level), which complicates the consensus problem.
In contrast with several solutions \cite{GUetal21, ZHENGetal23, XIAOetal22} that  rely on \emph{discontinuous} adaptive estimators and controllers, we propose a simple integral action which enables to achieve synchronization of all agents. 
  
\subsection*{Unmatched Disturbances}
  
The multi-agent system to be addressed is of the form:
  \begin{align}
  \label{sys3}
  \dot  x_i =& \; y_i  + d_i \quad x_i, y_i,\;  i \leq N \nonumber \\
  \dot y_i=& \; u_i 
  \end{align}
which clearly shows that the constant disturbance $d_i$ is {\it unmatched}, i.e., it cannot be directly cancelled by the input control $u_i$; The interconnection of the system is assumed to satisfy Assumption \ref{asp1}.
Similarly to the previous case, our goal is to design a robust controller such that the consensus can be established, i.e.:
\begin{align}
 \label{goal3}
\lim_{t\to \infty} &\quad |x_{i}(t)-\bar x_{m}(t) |=0, \;\; |\tilde y_{i}(t)- \bar y_{m}(t) |  =0 \nonumber \\
\lim_{t\to \infty} &\quad |\tilde \delta_{i}(t)-\bar \delta_m(t) |=0 
\end{align}
 with the so called {\it mean-field} dynamics  \cite{PANLOR17}
 \begin{align}
 \label{mf1}
 \bar x_{m}= v_\ell^\top x, \quad \bar y_{m}= v_\ell^\top \tilde y, \quad \bar \delta_m=v_\ell^\top \tilde \delta
 \end{align}
  which can be seen as a weighted average or equivalently 
\begin{align}
\label{goal4}
\lim_{t\to \infty} &\quad |x_{i}(t)-x_{j}(t) |=0, \;\; |\tilde y_{i}(t)-\tilde y_{j}(t) |  =0 \nonumber \\
\lim_{t\to \infty} &\quad |\tilde \delta_{i}(t)-\tilde \delta_j(t) |=0, \quad i \not=j, 
\end{align}  
with  $x= [x_{1}, \cdots, x_{N}]^\top$ $\in \rea^n$,  $\tilde y= [\tilde y_{1},  \cdots, \tilde y_{N}]^\top \in \rea^n$ and $ \tilde \delta = [\tilde \delta_1, \cdots, \tilde \delta_N]^\top \in \rea^n$, where  $\tilde y_{i} = y_{i}- k_s \hat \delta_i$ and  $\tilde \delta_i=k_s  \hat \delta +d_i$, for $\hat \delta (t)$ as a new state variable to be defined later with $k_s$ a free positive gain.

\begin{proposition}   \em
Consider the system \eqref{sys3} in closed loop with the controller
  \begin{align}
u_i= &\; -k_x \sum\limits_{j=1}^{\mathcal N} a_{ij} (x_i-x_j)-k_d \tilde y_i -k_s (  \alpha_1 x_i + \nu \tilde y_i) \nonumber  \\
\dot {\hat \delta}_i= &\; - \alpha_1 x_i- \nu \tilde y_i  
\label{intc3}
\end{align}
for positive control gains $k_x, k_s, k_d, \alpha_1, \nu$. 
From the extended vectors  $x$, $\tilde y$, $ \tilde \delta$ and defining $\bar y= [y_1+d_1, \cdots, y_N+d_N]^\top$, the closed-loop multi-agent system can be expressed in the following compact form:
\begin{align}
\label{tilsys3}
\dot x&= \bar y \nonumber \\
\dot {\bar y} &= -k_x L x-k_d \tilde  y  -k_s ( \alpha_1 x + \nu   \tilde  y)  \nonumber  \\
\dot {\tilde \delta} &=  -\alpha_1 x -\nu \tilde y 
\end{align}
with state synchronization error vectors
\begin{align}
\label{err3}
e_x=&\;(I_n-{\bf 1} v_\ell^\top) x, \quad e_y=\;(I_n-{\bf 1} v_\ell^\top) \tilde y \nonumber \\
e_d=&\;(I_n-{\bf 1} v_\ell^\top) \tilde \delta 
\end{align}  
This closed-loop system ensures the converge to zero of the dynamic consensus, i.e., $(e_x, e_y, e_d)$ $\to 0$, with a Lyapunov function defined as:
 \begin{align}
 \label{mW}
 {\mathcal W} =& \; \frac{1}{2}  \left[ \begin{array}{cc} e_x^\top  & e_y^\top  \end{array} \right]\left[ \begin{array}{cc} \alpha_1 P &  \nu P \\ \nu P & \alpha_2 I_n \end{array} \right]   \left[ \begin{array}{c} e_x \\ e_y  \end{array} \right] + \frac{1}{2} e_d^\top P e_d
 \end{align}
for a symmetric positive definite matrix $P=P^\top > 0$ satisfying $\sqrt{\alpha_1 \alpha_2 \over |P|}> \nu$ so that ${\mathcal W}>0$. \end{proposition}
 
\begin{proof}
By computing the time derivative of $e_x$, we obtain:
\begin{align}
 \dot e_x=& \; (I_n-{\bf 1}_N v_\ell^\top)(y+d)= (I_n-{\bf 1}_N v_\ell^\top) (y+d \pm k_s \hat \delta )  \nonumber  \\
  =& \;  (I_n-{\bf 1}_N v_\ell^\top) (y -  k_s \hat \delta + k_s \hat \delta + d )  \nonumber \\
  =& \;  (I_n-{\bf 1}_N v_\ell^\top) (\tilde y  + \tilde  \delta )  \nonumber \\
  =& \; e_y + e_d.   
  \label{dotexu}
  \end{align}
  By making use of the controller \eqref{intc3} and \eqref{err3}, the time derivative of $e_y$ yields:
  \begin{align}
   \dot e_y=& (I_n-{\bf 1}_N v_\ell^\top) (\dot y - k_s \dot {\hat \delta})\nonumber \\
   =& -(I_n-{\bf 1}_N v_\ell^\top)   \left( k_x L  x  -  k_d \tilde y  +  k_s ( \alpha_1 x + \nu \tilde y) + k_s \dot{\hat \delta} \right) \nonumber \\
=& - (I_n-{\bf 1}_N v_\ell^\top) (k_x L x +k_d \tilde y )   \nonumber \\
 =& -k_x (I_n-{\bf 1}_N v_\ell^\top)L(e_x+{\bf 1}_Nv_\ell^\top x) -k_d e_y  \nonumber \\
 =& -k_x L e_x -k_de_y, 
 \label{dotes}
\end{align}
 where to get the last equality we used the facts   $L {\bf 1}_N=0$ and $v^\top_\ell L=0$, and the definitions \eqref{err3}. 
 The time derivative of  ${\mathcal W}$ along  \eqref{dotes} and \eqref{dotexu} is given by:
 \begin{align}
 \dot{\mathcal W}=& \left[ \begin{array}{cc} e_x^\top  & e_y^\top  \end{array} \right]\left[ \begin{array}{cc} \alpha_1 P &  \nu P \\ \nu P & \alpha_2 I_n \end{array} \right]   \left[ \begin{array}{c} e_y+ e_d \\ -k_x Le_x -k_d e_y  \end{array} \right]   \nonumber \\
 & +e_d^\top P \dot e_d \nonumber \\
 =& \left[ \begin{array}{cc} e_x^\top  & \hspace{-.7mm} e_y^\top  \end{array} \right]    \left[ \begin{array}{c} \alpha_1 P(e_y+ e_d) - \nu P (k_x Le_x +k_d e_y)  \\ +\nu P(e_y+ e_d)  -\alpha_2(k_x Le_x +k_d e_y)  \end{array}  \hspace{-.6mm} \right]   \nonumber \\
 & +e_d^\top P \dot e_d \nonumber \\
 =&\; e_x^\top \left( \alpha_1 P-\nu k_d P-\alpha_2 k_x L^\top \right) e_y -e_y^\top \left( \alpha_2 k_d - \nu P     \right)e_y\nonumber \\
 & -\frac{1}{2} \nu k_x     e_x^\top (PL+L^\top P) e_x + \alpha_1 e_x^\top P e_d + \nu  e_y^\top P e_d \nonumber \\
  & +e_d^\top P \dot e_d \nonumber \\
   =&\; e_x^\top \left( \alpha_1 P-\nu k_d P-\alpha_2 k_x L^\top \right) e_y -e_y^\top \left( \alpha_2 k_d - \nu P     \right)e_y\nonumber \\
 & -\frac{1}{2} \nu k_x     e_x^\top e_x +e_d^\top  P(  \alpha_1   e_x + \nu  e_y+ \dot e_d)  
 \label{dotW}
  \end{align}
where to get the last equality we have invoked Lemma \ref{lem2} with $Q=I_n$.
The term $\dot e_d$ is computed by making use of the definition of the integral action \eqref{intc3} as follows: 
\begin{align}
\dot e_d=&\; (I_n-{\bf 1} v_\ell^\top) \dot{\tilde \delta}:=  -(I_n-{\bf 1} v_\ell^\top) ( \alpha_1 x +\nu \tilde y) \nonumber \\
=&\; -\alpha_1 e_x - \nu e_y.
\label{dotedu}
\end{align}
By replacing $\dot e_d$ into $\dot {\mathcal W}$ and setting $\nu= \frac{\alpha_1}{ k_d}$ we can express $\dot{\mathcal W}$ into the compact form: $\dot {\mathcal W}=-\frac{1}{2}e_{xy}^\top {\mathcal M} e_{xy}$, with an extended error vector $e_{xy}=\col(e_x, \; e_y) \in \rea^{2n}$ and a symmetric matrix defined as:
\begin{equation}
{\mathcal M}=  \left[ \begin{array}{ccc} \frac{\alpha_1}{k_d} k_x I_n &   \alpha_2 k_x L^\top  \\ 
 \alpha_2 k_x L &    2 (\alpha_2 k_d - \frac{\alpha_1}{k_d}P)
\end{array} \right]. 
\end{equation}
By setting $\alpha_1= k_d$ and computing the Schur complement of ${\mathcal M}$, we can show that ${\mathcal M}$ is positive semi-definite if the matrix satisfies:
$$
{\mathcal D}= 2\left(\alpha_2 k_d - P \right)- \alpha_2^2 k_xL L^\top\geq 0
$$
which can be easily ensured by defining the free control parameter $k_d>0$ such that: 
$$
k_d > \frac{1}{2} \alpha_2 k_x \lambda_L^2  +\frac{1}{\alpha_2} \lambda_P.
$$
With these conditions, we can ensure that $\dot {\mathcal W}\leq  - |{\mathcal M}|  | e_{xy}|^2$. 
Since ${\mathcal W}$ is positive definite and radially unbounded with respect to $e_x$, $e_y$ and $e_d$, it follows that { $(e_x, e_y)\in{\mathcal L}_2\cap\mathcal L_\infty$ } and $e_d \in\mathcal L_\infty$. 
This, in turn, implies that $\dot e_d \in\mathcal L_\infty$, thus, from the Barbalat Lemma, we conclude that $\lim\limits_{t\to\infty}e_d (t)=0$. 
This completes the proof. 
 \end{proof}

In the following Corollary \ref{cor2}, we prove the synchronization of all agents to a periodic (oscillatory) state. 
Similar to Corollary \ref{cor1}, we based our analysis on the weighted average dynamics. 

\begin{corollary} \label{cor2} \em 
Consider the mean-field coordinetes  \eqref{mf1} and the closed-loop dynamic \eqref{tilsys3}. 
Then, each agent converge to a periodic orbit governed by the following forced harmonic oscillator: 
\begin{equation}
\label{dynosc}
\left[ \begin{array}{c} \dot {\bar x}_{m} (t)   \\ \dot{\bar \delta}_{m} (t) \end{array} \right]:=\left[ \begin{array}{cc} 0  &  1 \\ -\alpha_1 & 0 \end{array} \right]   \left[ \begin{array}{c} {\bar x}_{m}(t) \\ {\bar \delta}_{m}(t)  \end{array} \right]  +\left[ \begin{array}{c} 1 \\1  \end{array} \right] {\bar y}_{m}(t)
\end{equation}
where the forcing term $\bar {y}_{m}(t)$ exponentially converges to zero, as it satisfies:
\begin{equation}
\label{x2m2}
\bar {y}_{m}(t)= \bar y_{m}(0)\exp^{-k_dt}.
\end{equation}
Moreover, 
\begin{equation}
\label{barx22}
\lim_{t \to \infty}  |\bar y_{i} (t) -  \bar \delta_m(t)|=0.
\end{equation}
is ensured by the closed-loop dynamic system.
\end{corollary}

\begin{proof}
The dynamic behavior of each agent is governed by the weighted average dynamics. 
This results from computing the time derivative of \eqref{mf1}:
\begin{align}
\dot {\bar x}_{m}=&\;  v_\ell^\top\bar y := v_\ell^\top \bar y  \pm v_\ell^\top k_s \hat \delta:=  \bar y_{m} + \bar \delta_m    \nonumber \\
\dot {\bar y}_{m}= &- v_\ell^\top \left( k_x L x +k_d  \tilde y  \right):= - k_d \bar y_{m}  \nonumber \\
\dot {\bar \delta}_m=& -v_\ell^\top (\alpha_1 x + \nu  \tilde y):= - \alpha_1 \bar x_{m}- \nu \bar y_{m}.
\label{dynav1}
\end{align}
Here, we have that the solution of the second equation is \eqref{x2m2}, which ensures that $y_{m}(t)$ converges exponentially to zero.
In turn, it implies that $y_{i} \to k_s \hat \delta_i $ exponentially. 
On the other hand, since $\bar {y}_{m}(t)$ converges exponentially to zero, the first and last equations of \eqref{dynav1} can be approximated as \eqref{dynosc}
 and their solutions clearly are continuous periodic signals. 
Finally, since  $y_{i} \to k_s \hat \delta_i $ exponentially and   $\tilde \delta_i (t) \to  \bar \delta_m(t)$, we have that 
\begin{align}
k_s \hat \delta_i (t)+d_i \to & \; \bar \delta_m(t)  \;\;\Rightarrow \;\; k_s \hat \delta_i (t) \to \bar \delta_m(t) -d_i
\end{align}
when  $t \to \infty$. Then, we can conclude that $y_{i}$ converges to the periodic signal $\bar \delta_m(t) -d_i$ or equivalently \eqref{barx22} holds.
This completes the proof.
 \end{proof}

\begin{remark} \em
 All the results can be extended to the case $x_i, y_i, u_i \in \rea^p$ with $p>1$ using the Kronecker product.
\end{remark}

\begin{remark} \em
In contrast with most robust controllers proposed in the literature (e.g., \cite{LVetal17, LIUWUZHO, WANGetal21}), in our case the disturbance rejection for matched disturbances is ensured with a simple continuous integral action. 
\end{remark}

\begin{remark} \em
The  proposed integral action has the same structure as the used in \cite{ROMNUNALD} to reject disturbances in a multi-agent system composed of nonholonomic robots. However, the Lyapunov function in that work is very complex compared with the formulated in this work. Moreover, the  interconnection in that work is described by  undirected graph.
\end{remark}

\begin{remark} \em
The presence of unmatched disturbances is very common in power systems, e.g., the DC-DC Buck power converter and the DC-AC converter, permanent magnet synchronous motors and inductor motor are systems. 
The generation of a resonant behavior as a sinusoidal signal in the voltages is a frequent task in these systems \cite{SAIISHISH, ANGOLITAB}.
\end{remark}

\begin{remark} \em
In contrast with classical solutions to the problem of multi-agents with unmatched disturbances, in this note we do not use discontinuous observers or discontinuous controllers. 
\end{remark}

 \section{Simulations}
 To validate the proposed dynamic consensus controller, in this section we carry out simulations using five agents.  According to Assumption \ref{asp1} the direct graph   is chosen as Fig. \ref{fig1}, where its corresponding Laplacian matrix appears on the right hand side of the figure. Moreover, to assess the robustness of the controller we also consider input disturbances subjected to {\it step changes} with a vanishing function, thus the  disturbances  commute  after 50 seconds from $d=\col( 0.1;-0.1;  0.2; -0.2; 0.1) + \frac{1}{12+t} $ to  $d=\col(0.2;-0.2;-0.1;0.2;-0.3) +\frac{1}{12+t} \exp(-0.2t)$.  The control gains were chosen as $\gamma_1=6$, $\gamma_2=17$, $\gamma_3=4$ and $b=10$.
  
 \subsection*{Matched Disturbances}
 
For this case, the control gains were chosen as $\gamma_1=6$, $\gamma_2=17$, $\gamma_3=4$ and $b=10$.
 
In  Figs. \ref{fig2} and \ref{fig3} we appreciate as all agents reach the consensus even in the presence of the time varying disturbances. Moreover,  Fig. \ref{fig4}  shows as the integral action $\hat \delta$  converges to the value $\frac{d(t)}{\gamma_3}$ as is predicted by the theory. For this,  notice that at time $t=50$ the term $\frac{1}{12+t}=0.0161$ so that $d(50)=\col(0.029, -0.021, 0.054, -0.046, 0.029) $, hence from the zoomed in of the left hand side of Fig \ref{fig4}  we see that $\hat \delta (50)= d(50)$.  On the other hand, after  50 sec, the disturbance commute to a different value where appears  the term $\exp(-0.2t)$. In this case we have that  at time $t=100$ that signal is equal to zero hence $\hat \delta (100) \to \col(0.05,-0.05, -0.025, 0.05, -0.075 ) $ as is appreciated in the zoomed in of the right hand side of Fig.  \ref{fig4}. Finally to corroborate the above discussion, in Fig. \ref{fig5} the converge to zero of $\tilde \delta$ is presented.  
  \begin{figure}[htp]
 \centering
\includegraphics[width=1\linewidth]{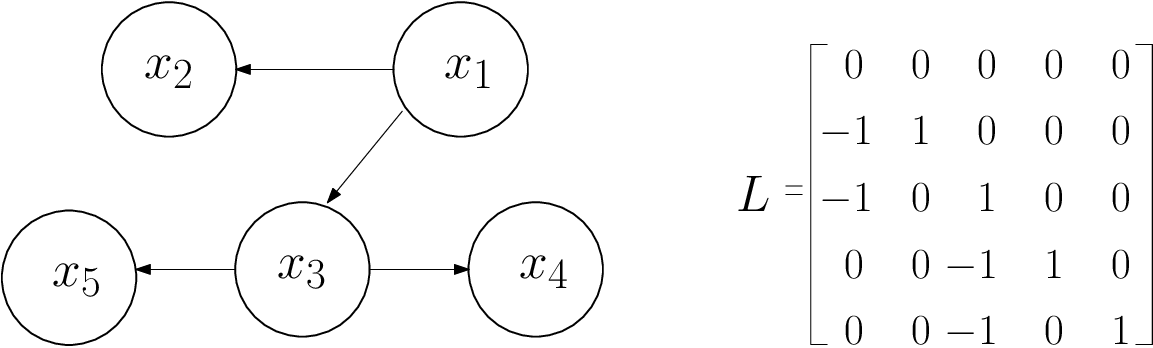}
\caption{A spanning-tree graph and its corresponding Laplacian}
\label{fig1}
\end{figure}
 
   \begin{figure}[htp]
 \centering
\includegraphics[width=1.05\linewidth]{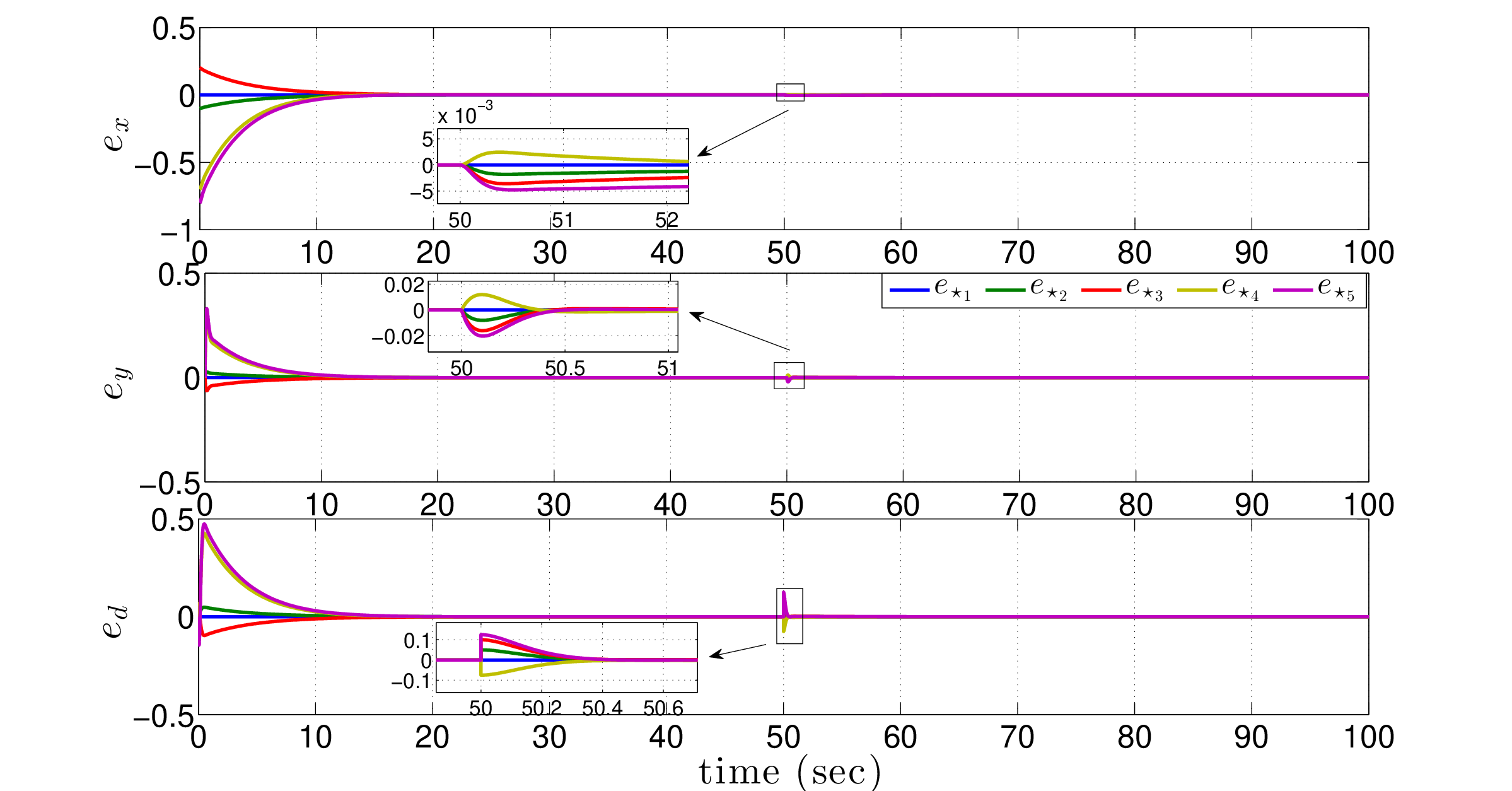}
\caption{Transient behavior of $e_{\star_i}$ with $i:1:5$}
\label{fig2}
\end{figure}

   \begin{figure}[htp]
 \centering
\includegraphics[width=1.05\linewidth]{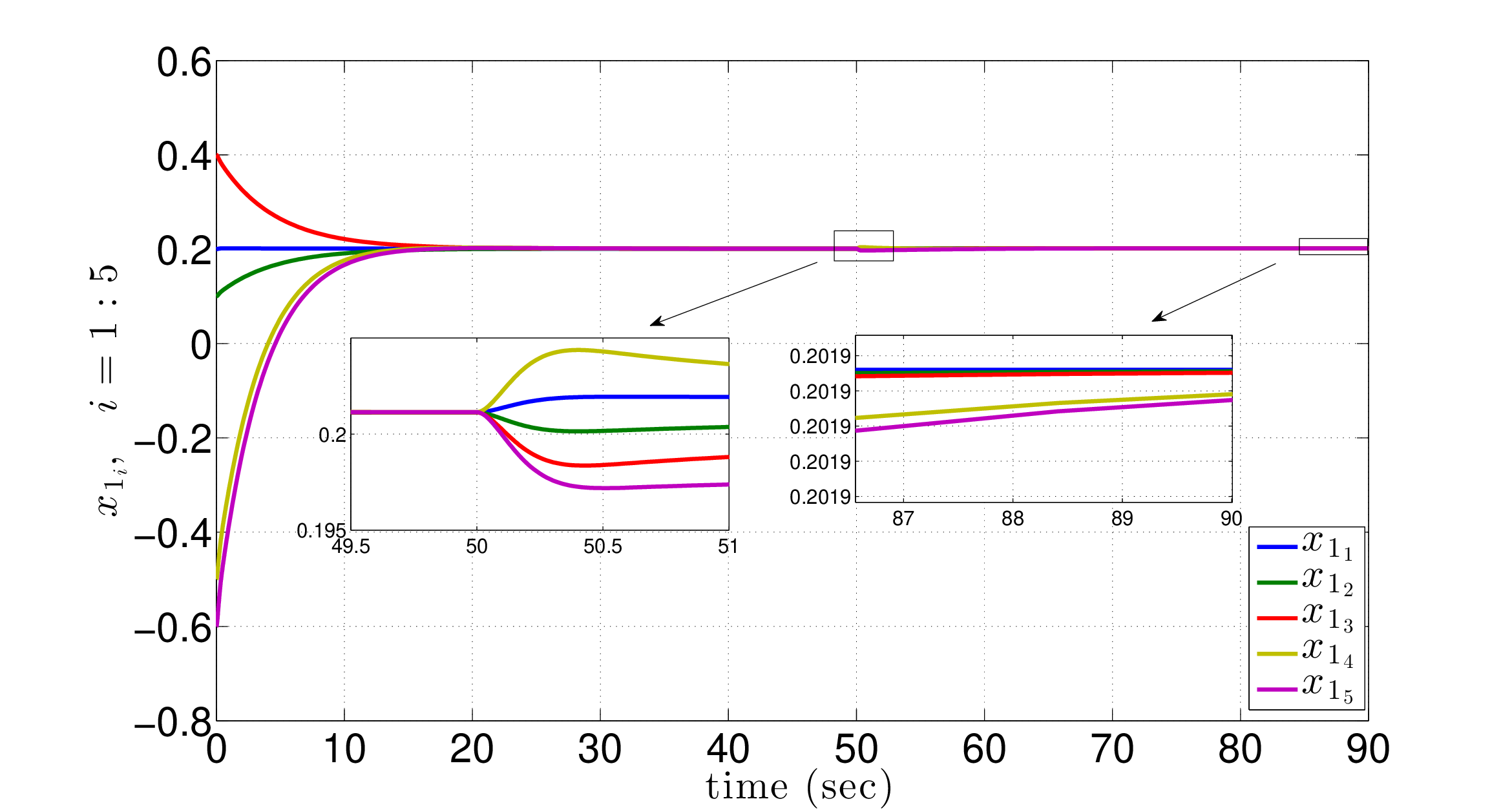}
\caption{Transient behavior of $x_i$ with $i:1:5$}
\label{fig2}
\end{figure}

   \begin{figure}[htp]
 \centering
\includegraphics[width=1.05\linewidth]{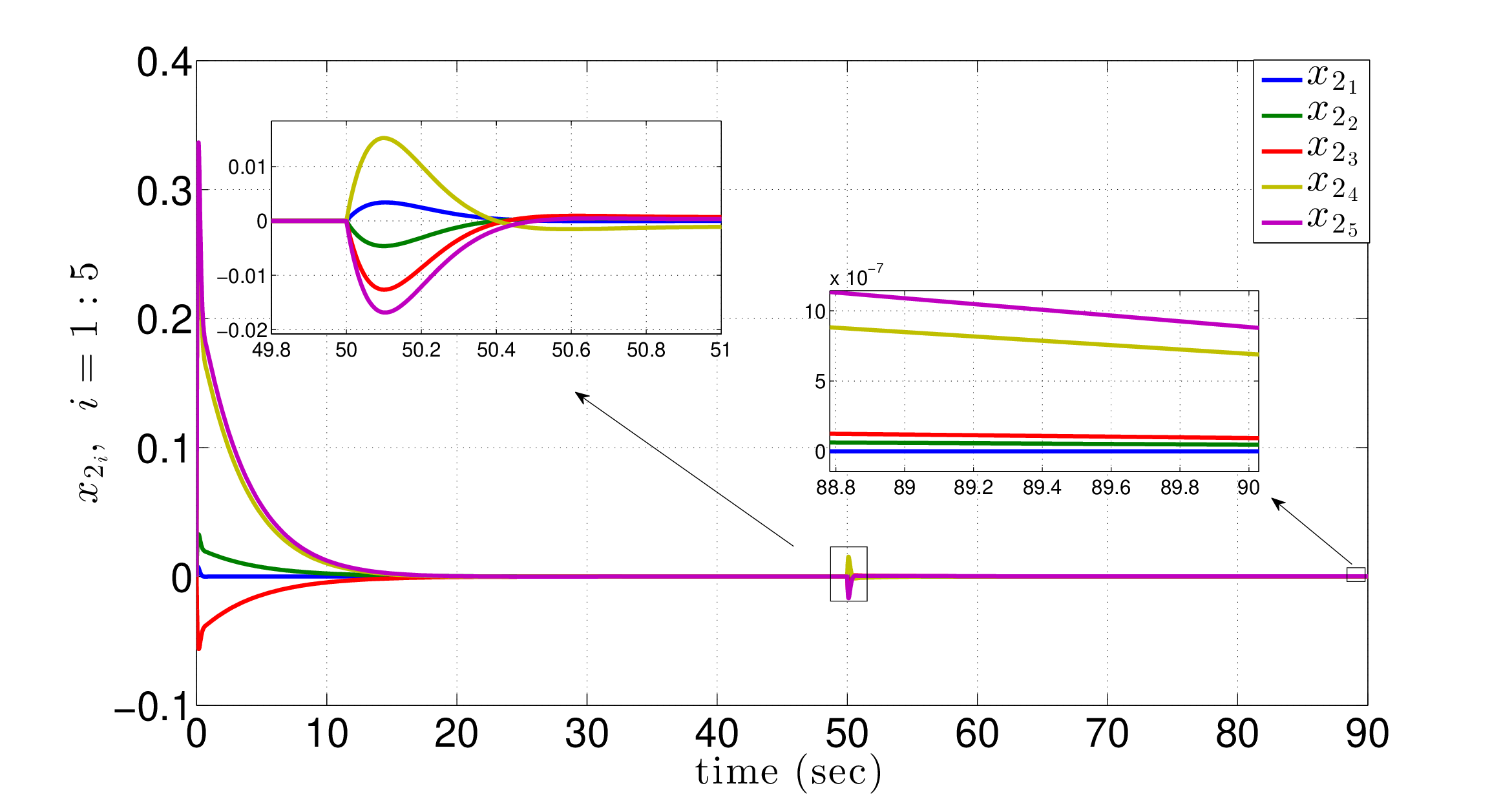}
\caption{Transient behavior of $y_i$ with $i:1:5$}
\label{fig3}
\end{figure}

     \begin{figure}[htp]
 \centering
\includegraphics[width=1.05\linewidth]{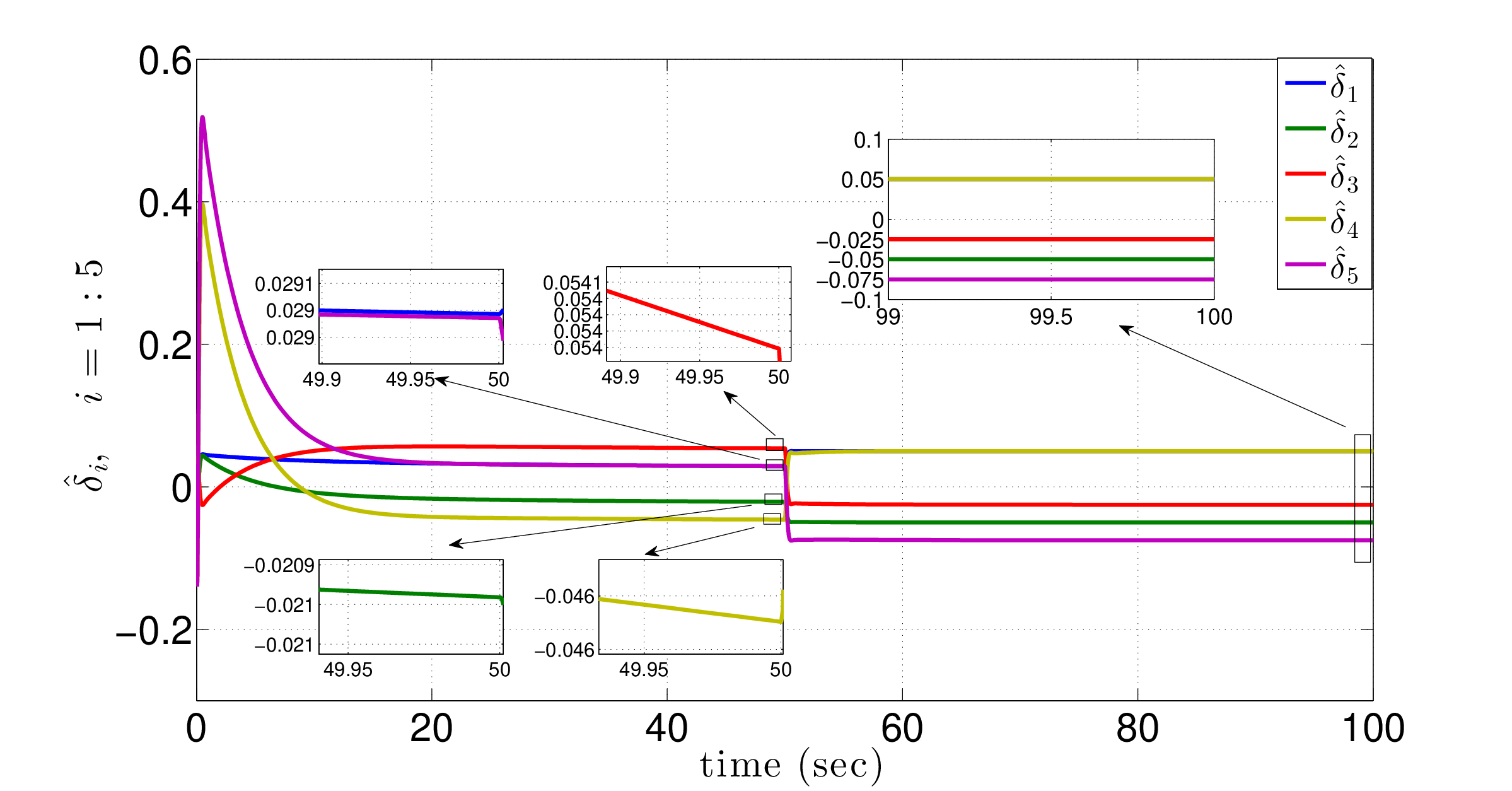}
\caption{Transient behavior of the integral action $\hat \delta_i$ with $i:1:5$}
\label{fig4}
\end{figure}

    \begin{figure}[htp]
 \centering
\includegraphics[width=1.05\linewidth]{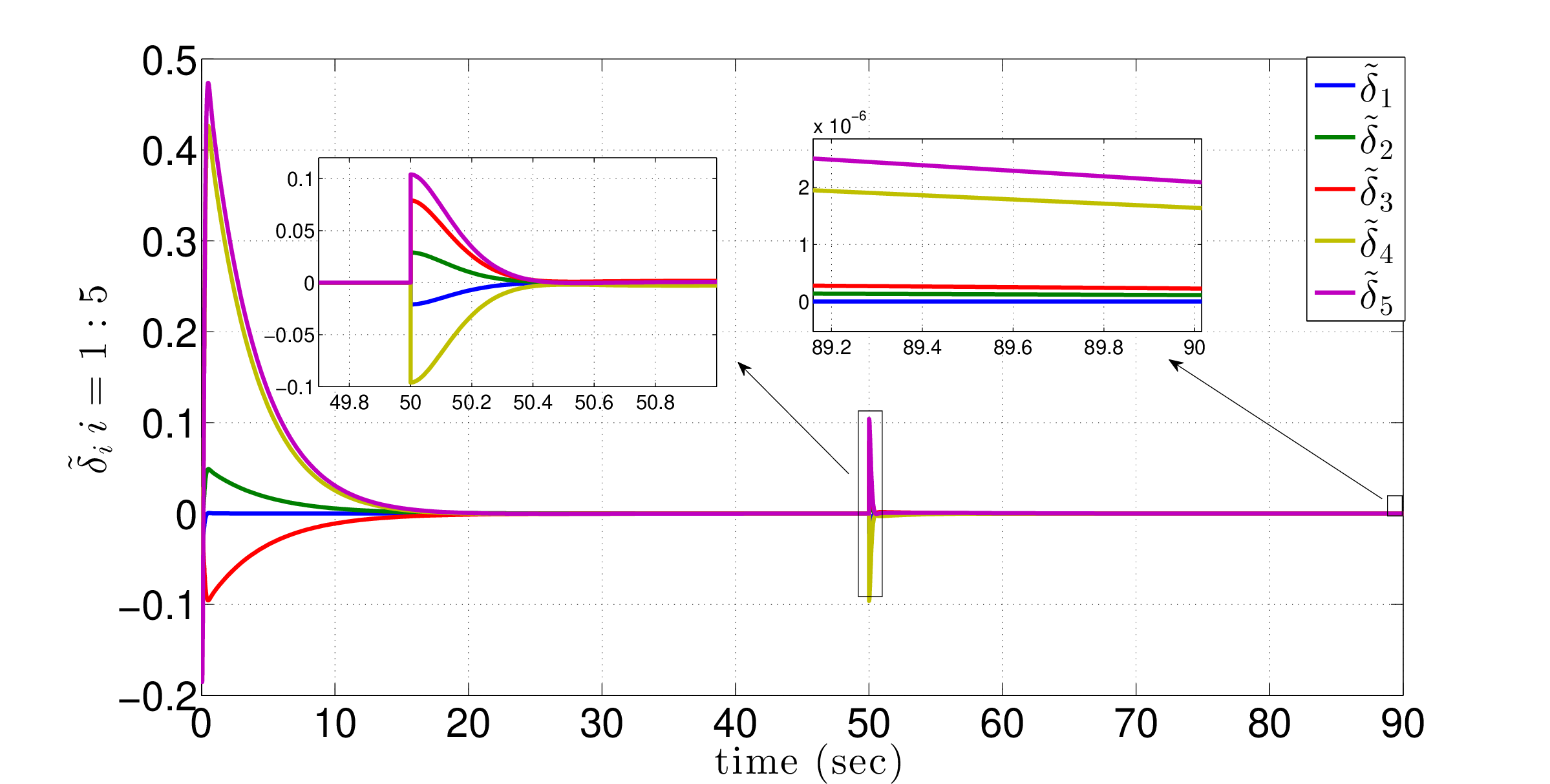}
\caption{Transient behavior of the error estimators $\tilde \delta_i$ with $i:1:5$}
\label{fig5}
\end{figure}

\subsection*{Unmatched Disturbances}
To carry out the simulations we used the same time varying disturbances as the matched case, but the disturbances commute after 20 seconds. The control gains were chosen as $\alpha_1=k_d=7.5$, $\nu=3$,  $k_s=5$ and $k_x= 3.4$. 

 Fig. \ref{fig7} shows the convergence to zero of the synchronization errors.  According to Corollary \ref{cor2}, in Fig. \ref{fig8} the periodic synchronization of all agents $x_{1i}$ (called {\it output signal})is ensured, as well as for  $\tilde \delta_{i}$, see  Fig. \ref{fig10} and convergence to zero of $\tilde \x_{2_i}$ in Fig. \ref{fig9}. Moreover, in contrast with all works addressing the consensus of unmatched disturbances, we notice in Fig. \ref{fig11} the periodic synchronization of $\bar x_{2_i}$. It is important to stress that the behavior of the agents does not suffer any important change under  the presence of the commutation of the disturbances. This bears out the robustness of our controller.

    \begin{figure}[htp]
 \centering
\includegraphics[width=1.05\linewidth]{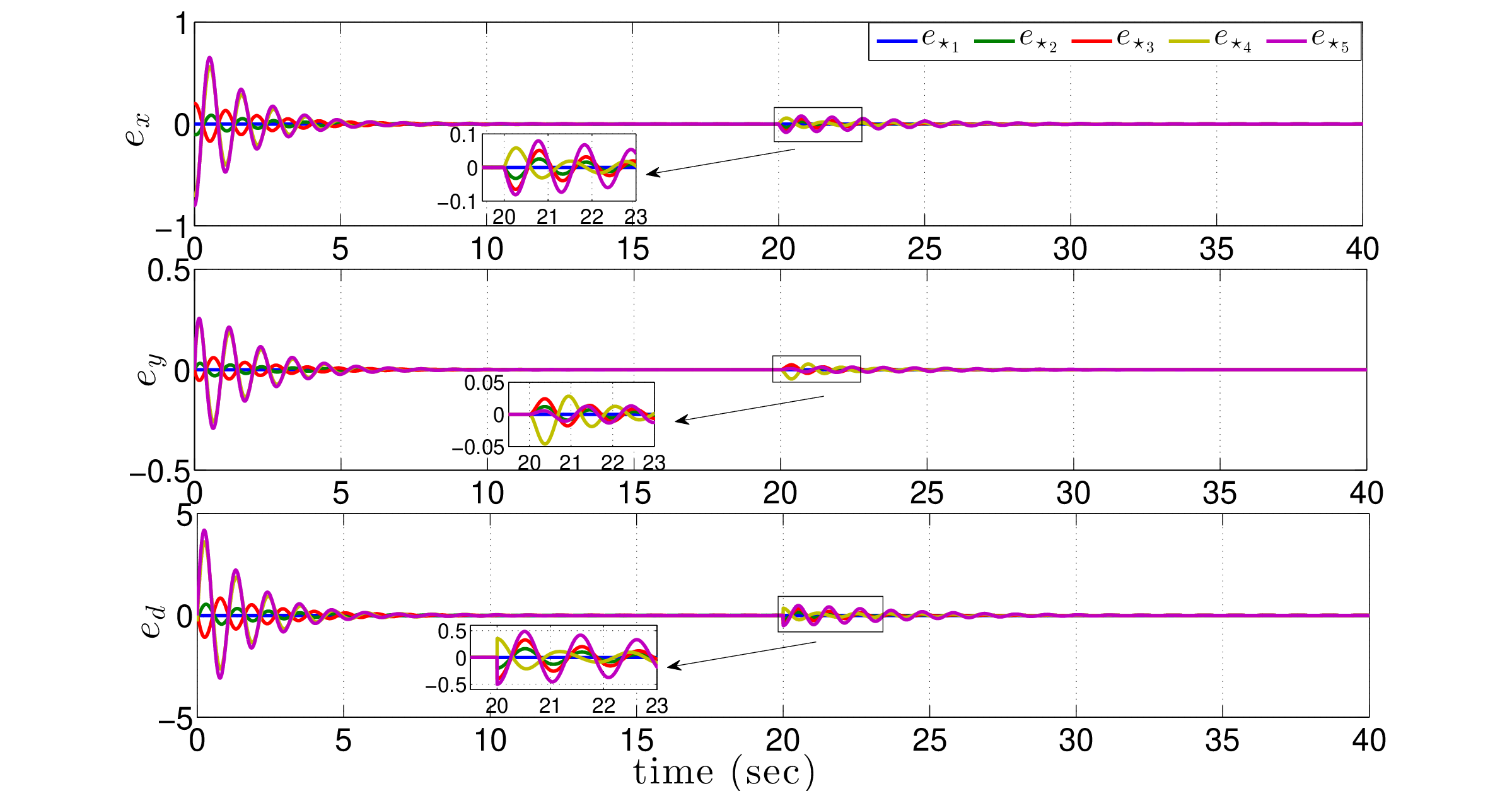}
\caption{Transient behavior of $e_{\star_i}$ with $i:1:5$}
\label{fig7}
\end{figure}

   \begin{figure}[htp]
 \centering
\includegraphics[width=1.05\linewidth]{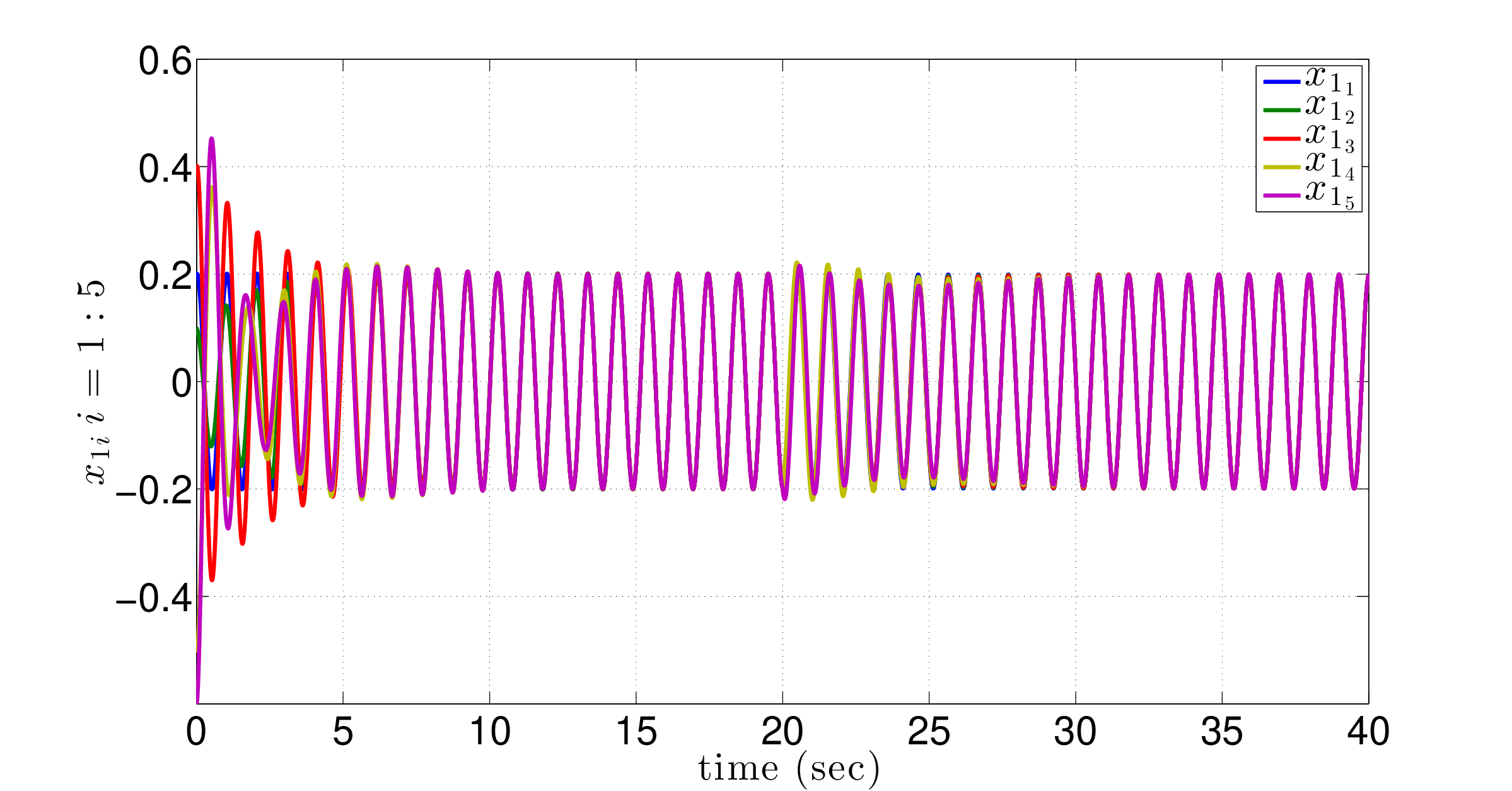}
\caption{Transient behavior of $x_{1i}$ with $i:1:5$}
\label{fig8}
\end{figure}

   \begin{figure}[htp]
 \centering
\includegraphics[width=1.05\linewidth]{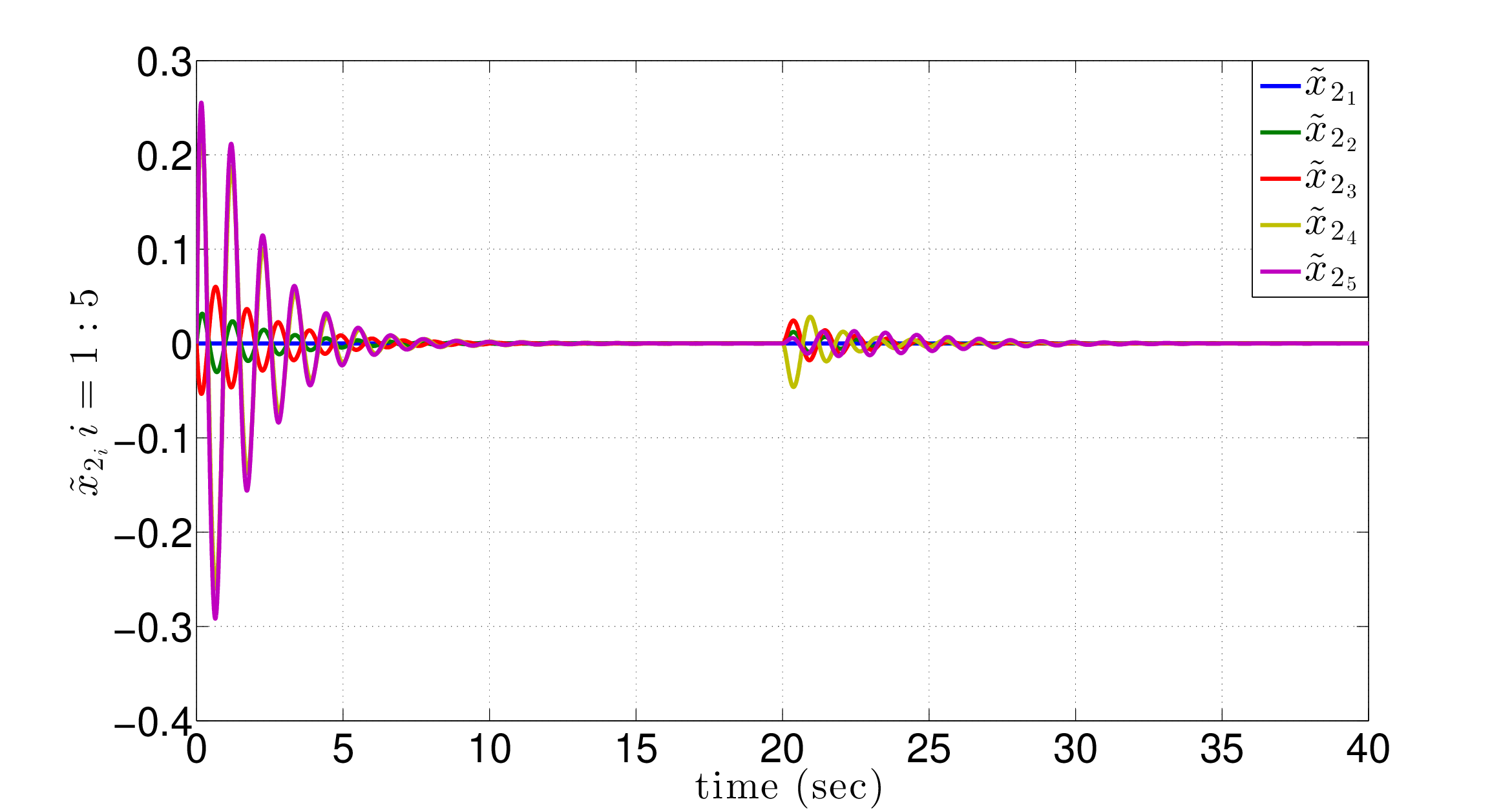}
\caption{Transient behavior of $\tilde x_{2_i}$ with $i:1:5$}
\label{fig9}
\end{figure}

     \begin{figure}[htp]
 \centering
\includegraphics[width=1.05\linewidth]{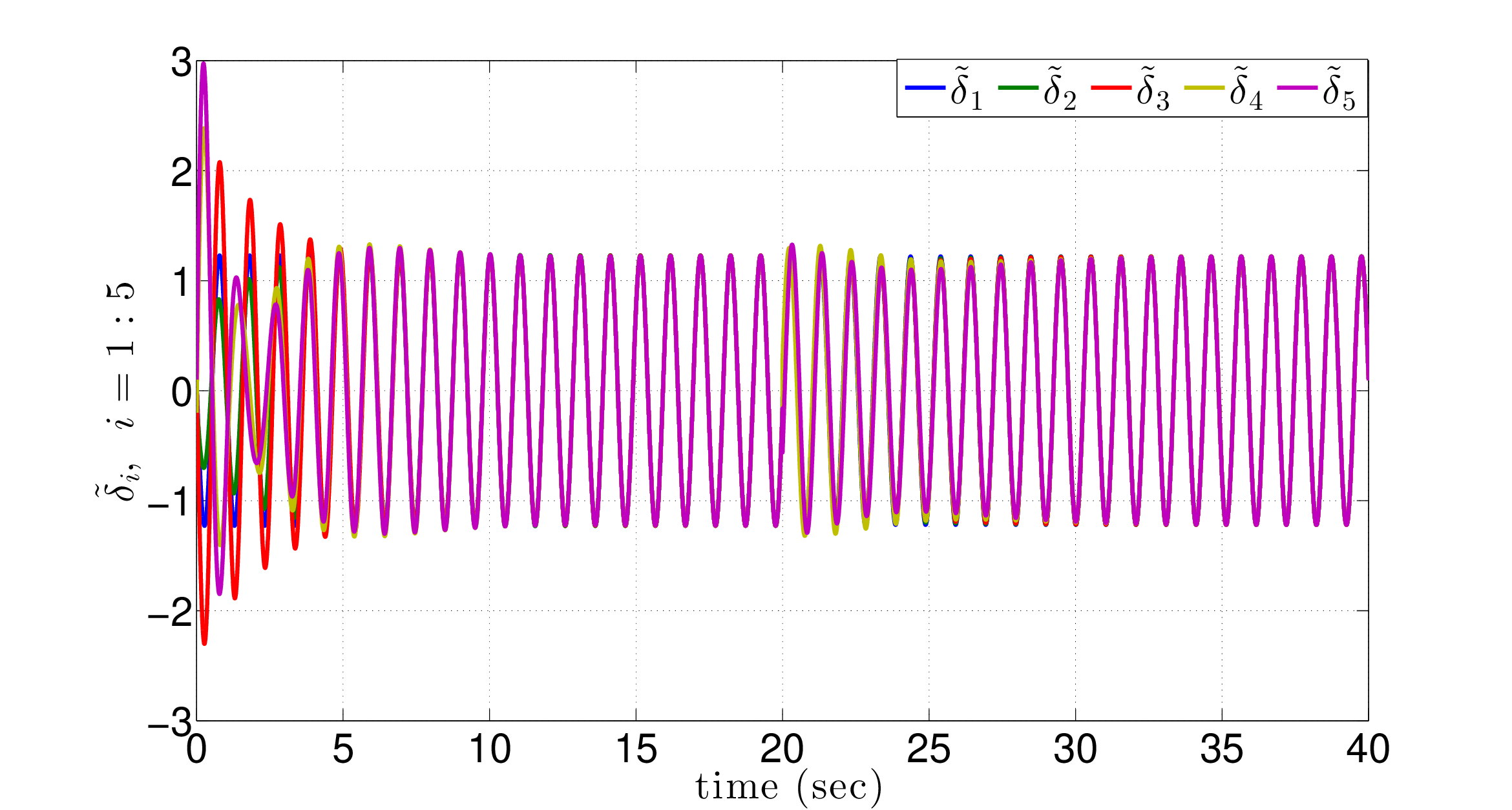}
\caption{Transient behavior of the auxiliary states $\tilde \delta_i$  with $i:1:5$}
\label{fig10}
\end{figure}

     \begin{figure}[htp]
 \centering
\includegraphics[width=1.05\linewidth]{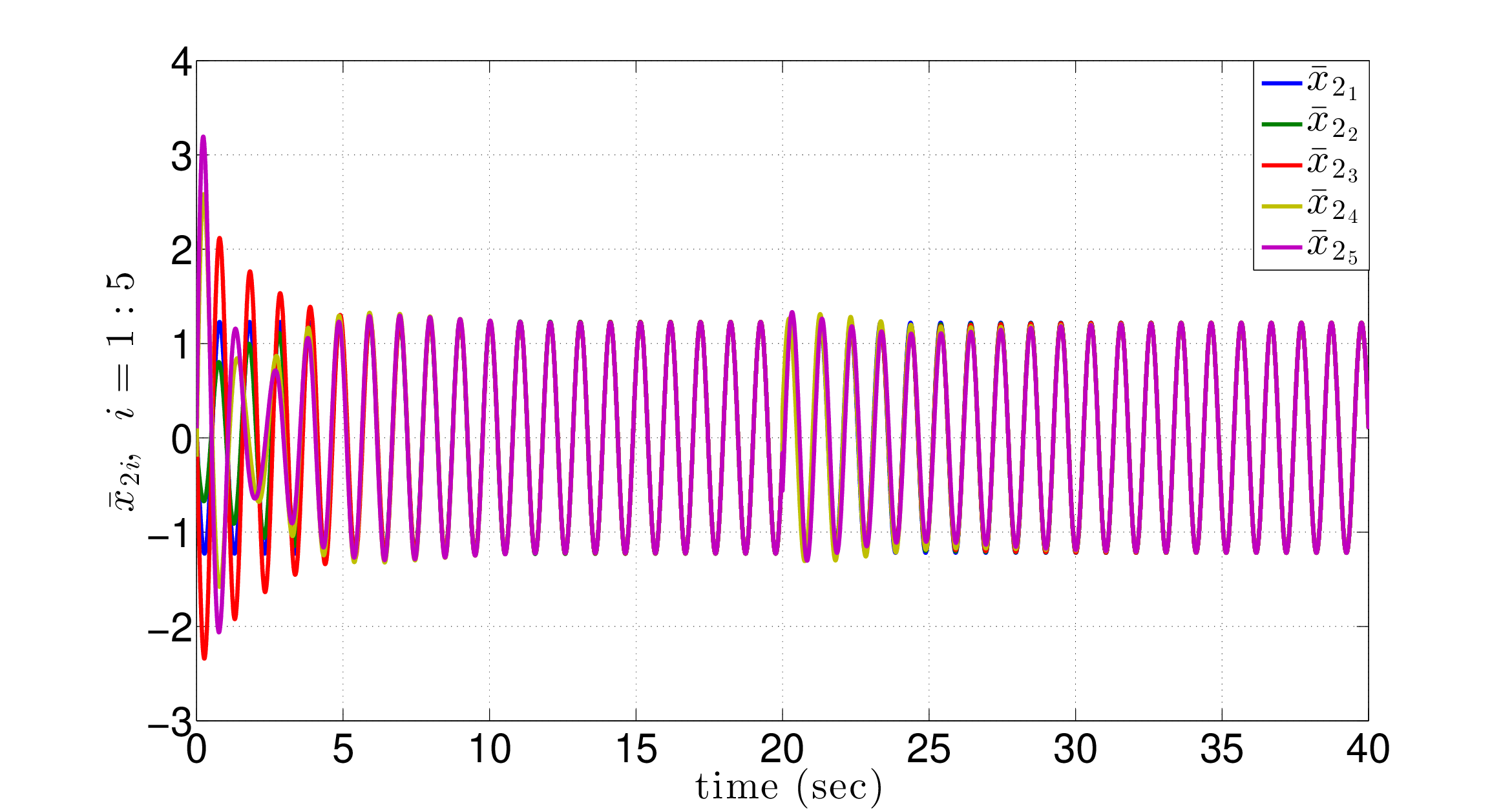}
\caption{Transient behavior of the states  $\bar x_{2_i}$ with $i:1:5$}
\label{fig11}
\end{figure}

\section{Conclusions}
In this work, we have presented two simple continuous controllers to ensure consensus and synchronization of perturbed double integrator systems interconnected under a directed graph containing a spanning tree. 
When matched disturbances are considered, the proposed method resembles a Proportional-Integral- Derivative (PID) controller, whose new integral action enables the rejection of disturbances. On the other hand, for the synchronization problem, an integral action handles the unmatched disturbances without the use of high-gain and discontinuous techniques. The stability of problems is formally proved with strict Lyapunov analysis. 
 

\begin{IEEEbiography}[{\includegraphics[width=1in, height=1.25in, clip, keepaspectratio]{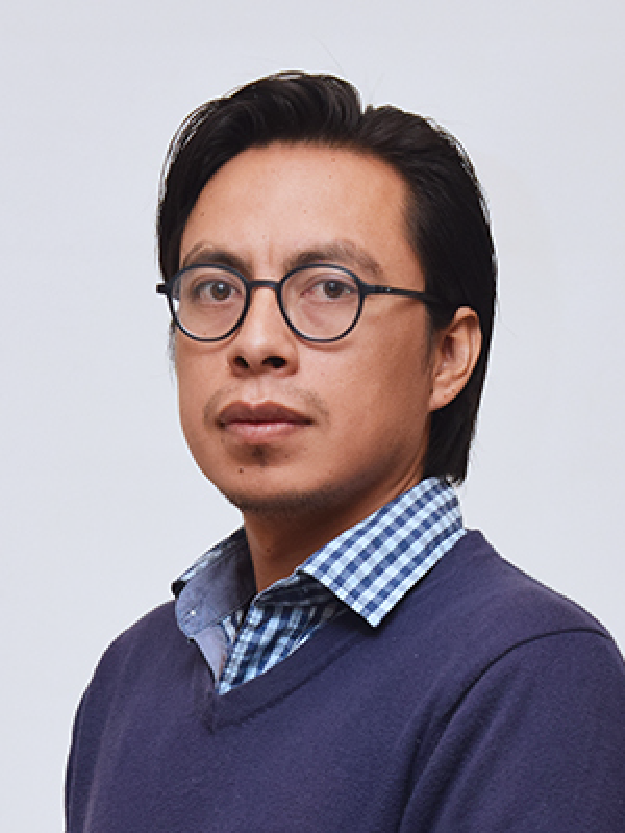}}]
{Jose Guadalupe Romero} (Member, IEEE)  obtained the Ph.D. degree in Control Theory from the University of Paris-Sud XI, France in 2013. 
Currently, he is a full time Professor at ITAM in Mexico and since 2023 he is  the Chair of the Department of Electrical and Electronic Engineering. He has over 45 papers in peer-reviewed international journals where he has also served as a reviewer.
His research interests are focused on nonlinear and adaptive control, stability analysis  and the state estimation problem, with application to mechanical systems, aerial vehicles, mobile robots and multi-agent systems. He currently serves as an Editor of the \textsc{International Journal of Adaptive Control and Signal Processing}.
\end{IEEEbiography}

\begin{IEEEbiography}
[{\includegraphics[width=1in,height=1.25in,clip,keepaspectratio]{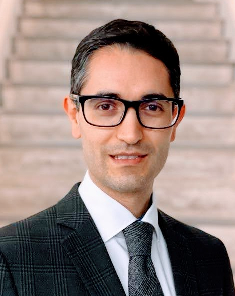}}] {David Navarro-Alarcon} (Senior Member, IEEE) received the Ph.D. degree in mechanical and automation engineering from The Chinese University of Hong Kong, in 2014. 
Since 2017, he has been with The Hong Kong Polytechnic University, where he is currently an Associate Professor with the Department of Mechanical Engineering, and the Principal Investigator of the Robotics and Machine Intelligence Laboratory.
His current research interests include perceptual robotics and control systems.
He currently serves as an Associate Editor of the \textsc{IEEE Transactions on Robotics (T-RO)} and Guest Associate Editor of the \textsc{Journal of Field Robotics}.
\end{IEEEbiography}

\end{document}